\documentclass[11pt]{article}

\usepackage{amssymb}
\usepackage{amsfonts}
\usepackage{amsmath,amsthm}
\usepackage{graphics}
\usepackage{graphicx}

\newtheorem{lemma}{Lemma}

\newtheorem{definition}{Definition}
\newtheorem{proposition}{Proposition}

\def\EE{\mathord{I\kern-.35em E}}
\begin{document}
\smallskip

\thispagestyle{empty}
\bigskip
\bigskip
\bigskip
\bigskip
\bigskip
\bigskip
{\setlength{\parskip}{0.6cm}
\renewcommand{\thefootnote}{\fnsymbol{footnote}}
\begin{center}
{\LARGE \bf Illiquid Financial Markets \\ and Monetary Policy\footnote[1]{$\
    $We are very grateful to Darrell Duffie, Ricardo Lagos, and Randall Wright for their guidance and helpful discussions. We would also like to thank Philipp Kircher, Jan Eeckhout, Aus\'{i}as Rib\'{o}, and seminar participants at the Summer Workshop on Money, Banking, Payments and Finance at the Federal Reserve Bank of Chicago, the Search and Matching Workshop at University of Pennsylvania, Universidad P\'{u}blica de Navarra, Universitat Aut\'{o}noma de Barcelona, and Universitat de Barcelona. Financial support from the Spanish Ministry of Science and Innovation through grant ECO2009-09847, the Juan de la Cierva Program, and the Barcelona Graduate School Research Network is gratefully acknowledged. All errors are solely our responsibility.}}

\bigskip
\smallskip

{\large Athanasios Geromichalos\footnote[2]{$\
    $Department of Economics, UC Davis. Contact: ageromich@ucdavis.edu}, Juan M. Licari\footnote[3]{$\
    $Moody's Analytics. Contact: juan.licari@moodys.com}, Jos\'{e} Su\'{a}rez-Lled\'{o} \footnote[4]{$\
    $Departmento de Econom\'{i}a e Historia Econ\'{o}mica. Universidad Aut\'{o}noma de Barcelona. Contact: jose.suarezlledo@uab.es}}
\end{center}}

{\setlength{\parskip}{0.3cm}
\begin{center}

\bigskip
\smallskip

    July, 2011
\end{center}}

\renewcommand{\thefootnote}{\arabic{footnote}}

\medskip

\begin{abstract}

\noindent This paper analyzes the role of money in markets where financial investment takes place in a decentralized fashion. A key methodological contribution is the development of a dynamic framework that brings together a model for illiquid financial assets \textit{\`{a} la} Duffie, G\^{a}rleanu, and Pedersen, and a monetary framework \textit{\`{a} la} Lagos and Wright. The presence of decentralized financial markets generates an essential role for money in helping investors re-balance their portfolios. From the equilibrium conditions we are able to derive an asset pricing theory that delivers an explicit connection between monetary policy, the different asset prices, and welfare. In contrast to the existing monetary literature, we can sustain welfare gains even with positive inflation rates. Also, we obtain a negative relationship between inflation and asset prices, that are always above their fundamental value. This price differential is associated with a liquidity premium.

\end{abstract}

\bigskip
\medskip


\medskip


\newpage
%%%%%%%%%%%%%%%%%%%%%%%%%%%%%%%%%%%%%%%%%%%%%%%%%%%%%%%%%%%%%%%%%%
%%%%%%%%%%%%%%%%%%%%%%%%%%%%%%%%%%%%%%%%%%%%%%%%%%%%%%%%%%%%%%%%%%%
\section{Introduction}

Over the last two decades secondary financial markets have developed considerably, both in size and complexity. New financial products are generated through processes of securitization by which financial assets are derived from the value of an underlying asset. In their seminal paper Duffie, G\^{a}rleanu, and Pedersen (2005) document that many of these financial products are traded in markets that are characterized by search frictions. At the same time, since the influential work of Lagos and Wright (2005) monetary search models have been integrated into main-stream macroeconomic theory. Their model analyzes the role of monetary exchange in economies where trade is not centralized through some perfect and frictionless (Walrasian) market.

Our paper attempts to bring these two strands of the literature together by studying the role that money can play, through liquidity provision, in frictional financial markets. Duffie et al. analyze Over-the-Counter markets where financial assets trade without the use of a liquid asset like fiat money. On the other hand Lagos and Wright (2005) consider a model where money helps overcome certain frictions of decentralized trade. However, agents here trade consumption goods for the liquid asset. Our objective is to bring the framework in Duffie et al. (2005) into a dynamic monetary model where financial markets with frictions may generate a role for money.\footnote{\,As it is well known, in models that feature Arrow-Debreu type of markets it is very hard to support monetary equilibria other than those where agents are \symbol{92}forced\symbol{34} to hold money, e.g. money in the utility function or cash-in-advance models.} We present a model with a sequence of markets where agents can store a safe-return real asset and later on convert part of it into a risky and illiquid financial investment. The modeling of this risky asset follows Duffie et al.: it yields different idiosyncratic returns and agents can re-balance their portfolios in illiquid (decentralized) markets. A relevant innovative element in our approach is that investors can (and will) use money as a medium of exchange for those transactions. We show that money is not neutral and has an impact on both the volume of trade and the value of assets.

Duffie et al. consider the specific trading structure in decentralized financial markets where assets are traded by procedures like bid-ask pricing, bilateral or multilateral bargaining, etc. They build a model that rationalizes standard measures of liquidity in these markets such as trade volume, bid-ask spreads, and trading delays. However, investors' ability to re-balance their positions is severely limited by the restriction that agents can hold either 0 units or 1 unit of the assets. On the other hand, the search-theoretic literature stemming from Lagos and Wright (2005) develops a framework where centralized and decentralized markets interact. We consider this an appropriate setup for analyzing these issues that has not been exploited in our direction.

Lagos (2011) presents an asset pricing model where equity shares and fiat money can be used as means of payment for consumption, and derives a monetary policy scheme as a function of stochastic equity returns. In this model, however, there is neither collateral nor trading of assets, just exchange of consumption goods for real assets and fiat money. In a paper that more specifically tries to approach Over-the-Counter markets in a less restrictive manner than Duffie et al., Lagos and Rocheteau (2009) develop a framework where investors may wish to adjust their asset positions after a shock to preferences that affect their asset demand. They may do so contacting financial intermediaries that charge a fee. The authors try to explain the above mentioned dimensions of liquidity in these markets. In a more recent paper, and within a very similar setup, Lagos, Rocheteau, and Weill (2011) take a step further and consider crises and recoveries in these markets. They define a crisis as a negative shock to investors' asset demand that will recover at a random time. Liquidity can be provided by dealers, since they accumulate asset inventories. In the cases where it is not optimal for the dealers to accumulate assets, the government can step in and provide liquidity instead. Nevertheless, these models have neither fiat money nor assets serving as collateral, and therefore their notions of liquidity differ somehow from the standard ones. Ferraris and Watanabe (2011) model capital accumulation that will serve as collateral, and discuss the effects of fluctuations in the liquidation value of collateral on the economy's allocation and its interaction with monetary policy. As it turns out, however, we find that no relevant paper in this field has yet tried to formally approach the integration of money in illiquid financial markets.

Thus, we want to build a model that fills the gap between the finance and monetary literature. Our main contribution is the development of a framework that allows us to derive an asset pricing theory that delivers an explicit connection between monetary policy, the different asset prices, and welfare. We find that the value generated in these markets is affected by monetary policy, and the equilibrium price of the underlying asset is always above its fundamental value. In particular, two of our results stand in contrast to those in most of the search literature. First, under certain parameterizations welfare gains can be sustained even by positive inflation rates. In the related literature the optimal policy usually depicts deflationary money growth rates. Second, our pricing equations reveal that the monetary policy instrument and the price of the underlying asset are negatively correlated. The reason is that money and the asset are complements in our model, as opposed to a majority of the literature where they are substitutes. This aspect is further elaborated by Ferraris (2010) and Geromichalos, Licari, and Su\'{a}rez-Lled\'{o} (2007). We shall later discuss our results in more detail.

The rest of the paper is organized as follows. In sections 2 and 3 we present the main features of the model and describe the market structures that will allow us to study the role of the different assets, which is derived from the optimal behavior of each type of agent. In section 4 we proceed to analyze the existence of a monetary equilibrium as well as the nature of asset prices and the welfare consequences of inflation. Section 5 contains the main conclusions. In an Appendix we also consider two alternative pricing protocols in the decentralized market to test the robustness of the results.

\section{The Model}

The environment that we analyze takes the framework presented in Lagos and Wright (2005), henceforth LW, as a starting point. Time is discrete and there is a $[0,1]$ continuum of agents that live forever and discount future at rate $\beta \in (0,1)$. There are three assets in the model: a real asset that yields a deterministic return, $R$, at the end of every period, a risky financial asset, and an intrinsically worthless object that we call money. In every period agents engage in different activities in three markets that open sequentially. At the beginning of a period, $t$, agents enter an \textit{investment market} (IM) with a certain amount of the real asset, $a_t$, and money holdings, $m_t$, that they have previously stored. In this first market agents can decide how much of the safe asset to transform into a risky financial asset, $s_t\in[0,a_t]$, thus giving up any claim on the safe real return by the issued amount. Instead, investing in this new asset will yield a high return, $y_{_H}$, or a low return, $y_{_L}$, at the end of the period with equal probability. Furthermore, in order to make the analysis more interesting we are going to assume that $0\leq y_{_L}<y_{_H}$, and that $y_{_H}>R$. Following Duffie et al., we model idiosyncratic returns such that depending on the type of investor I turn out to be I will obtain different return. However, this idiosyncratic uncertainty will only be resolved right after the investment decisions have been made and before entering a \textit{decentralized financial market} (DFM). These returns are i.i.d. distributed across periods and agents. \newline \indent In the second market, having learned their types, agents can readjust their portfolios by re-balancing their holdings of the risky asset. This process happens in bilateral meetings where agents' types are public info. In particular, gains from this decentralized trade are generated \textit{\`{a} la} Berentsen and Rocheteau (2003) from low-return agents wanting to sell their investment assets to high-return agents. Anonymity, in the sense that trading histories are not known among agents, is still present in this environment. Also, there is imperfect enforcement and a double coincidence problem.\footnote{\,Only agents of different types will be willing to trade their risky assets: high types want to take on larger investment amounts in assets that will give them higher returns than the safe asset; and low types may want to get ride of low-yield investment.} These are some of the main elements that render money essential in this type of framework, as discussed in Kocherlakota (1998). In particular, we will show that money will allow to reallocate otherwise inefficient investment allocations. \newline \indent Finally, after the investment and the decentralized financial market everyone enters a \textit{centralized market} (CM) with their new holdings of the risky asset, the safe asset, and money. At this last stage agents will choose how much of the safe real asset and money they will carry onto the next period. They also derive net utility, $U(X_t)-H_t$, from consuming an amount of a general good, $X_t$, and from supplying an amount of labor, $H_t$. Notice that this is the only market in which agents consume and work. Utility, $U(X_t)$. is assumed to be twice continuously differentiable with $U'>0$, $U''\leq 0$. We also assume there exists $X^*$ such that $U'(X^*)=1$, with $U'(X^*)>X^*$. At this stage, agents can acquire any quantity of money, $\hat{m}_t$, and the real asset, $\hat{a}_t$, for the next period at prices $\phi$ and $\psi$, respectively. The supply of this asset is fixed, $A$, and each unit yields dividend, $R$, in the last market of the next period. Money supply is controlled by a monetary authority and it evolves according to $\hat{M}_t=(1+\mu)M_t$. \newline \indent This paper departs from LW and other related papers in two fundamental aspects that will be discussed in detail below: one is the introduction of an investment market for securitization, and the other relates to how we model decentralized financial trade. As we mentioned in the introduction, this departure will generate very interesting results, also regarding monetary policy. Some stand in contrast to those in related papers. Finally, it is an important issue in the literature that deals with decentralized trade, and particularly in the money-search literature, how different bargaining protocols may determine the terms of trade. A broad discussion on this is provided by Rocheteau and Wright (2005). In our paper we consider both price taking and bargaining.\footnote{\,Price taking can be regarded as the monetary version of Lucas and Prescott (1974).} However, the structure that we endow decentralized trade with turns out specially relevant. Indeed, we find that Arrow-Debreu equilibria where the price of the real asset deviates from its fundamental value can only be supported under price taking. Therefore, event though the bargaining version has some interesting aspects, we relegate its analysis and the justification of this result to the Appendix. In order to save notation, we drop the time subscripts from now on whenever it does not lead to confusion.

\section{Optimal Behavior}
\subsection{Centralized Market}

In order to derive a more clear intuition of the mechanism of the model, and for convenience of analysis, we proceed to solve the model backwards, starting from the third sub-period, a Walrasian centralized market. The value function of entering the CM with money holdings $m$, real asset holdings $b$, and financial investment $s$ is
\begin{eqnarray}
V^{3j}(m,b,s)=\max_{\hat{m},\hat{a},X,H}\left\{U(X)-H+\beta V^1(\hat{m},\hat{a}) \right\}\nonumber\\
\textit{s.t. } \phi \hat{m}+\psi\hat{a}+X=H+\phi m+(\psi+R)b+(\psi+y_{_j})s, \nonumber
\end{eqnarray}where $j=L,H$, and $b=a-s$, with $a$ being the amount of the real asset brought from the previous period. Three observations are immediate. First, in every period $X=X^*$ and we can write
\begin{eqnarray}
V^{3j}(m,b,s)=U(X^*)-X^*+\phi m+(\psi+R)b+ \nonumber \\ (\psi+y_{_j})s+\max_{\hat{m},\hat{a}}\left\{-\phi\hat{m}-\psi\hat{a}+\beta V^1(\hat{m},\hat{a}) \right\}.\label{V3}
\end{eqnarray}\noindent Second, $V^{3j}$ is linear in all its arguments,
\begin{eqnarray}
V^{3j}(m,b,s)=\Lambda+\phi m+(\psi+R)b+ (\psi+y_{_j})s,\label{linearity}
\end{eqnarray}where the definition of $\Lambda$ is obvious. Finally, it is easy to see from (\ref{V3}) that there are not any wealth effects: the agent's choices of $\hat{m},\hat{a}$ do not depend on today's states $m,a$. This is a consequence of quasi-linearity of preferences. Although we will discuss the solution to the optimization problem above later, we lay out here the first order conditions for $\hat{m}$ and $\hat{a}$, \begin{eqnarray}-\phi + \beta V^1_m(\hat{m},\hat{a}) \leq 0 && (=0,\textrm{ if } \hat{m}>0 \label{foc mhat}), \\ -\psi + \beta V^1_a(\hat{m},\hat{a}) \leq 0 && (=0,\textrm{ if } \hat{a}>0), \label{foc ahat}\end{eqnarray} where $V^1_i$ stands for the partial derivative of $V^1$ with respect to its argument $i$. From expression (\ref{linearity}) the envelope conditions would be \begin{eqnarray} V^{3j}_m=\phi m;\,\,\,V^{3j}_b=\psi + R;\,\,\,V^{3j}_s=\psi + y_j. \nonumber \end{eqnarray} The first order conditions refer to the fact that the cost of carrying the assets across periods must not be negative. Otherwise agents would want to accumulate unbounded amounts of money and the real asset. In a pure Arrow-Debreu world the equilibrium condition for the real asset would imply its price just reflecting the discounted flow of returns. In our model its price will reflect the discounted marginal value of carrying the asset into the next period, which will differ from its fundamental value.

\subsection{Decentralized Financial Market}

A large percentage of the financial sector is composed of markets for less liquid assets that are more difficult to trade in regular centralized markets. Trade is carried out in more decentralized procedures in these markets. Over-the-Counter markets are a especially relevant example (Duffie et al. (2005), Lagos et al. (2010)). However, it is also the case that in many of these markets agents do not usually bargain. Instead, they feature some type of price taking protocol (bid-ask pricing). Following the logic of these type of markets, we model here the exchange of a risky financial asset, $s$, that is less liquid than a standard real asset, $a$, commonly traded in a centralized Walrasian market. We also consider a price taking mechanism that resembles those in the kind of markets we are interested in. \newline \indent When agents enter our DFM they know what return they will obtain for their investment in the risky asset, and they are now allowed to re-balance their positions on this investment by trading the risky asset. Since $0<y_{_L}<y_{_H}$, L-types will naturally arise as the sellers of $s$, while the H-types will naturally become the buyers.\footnote{\,In principle it makes sense that L-types want to get rid of a bad asset, while H-types want to purchase more of an asset on which they can get a high return. As will be shown, the amount of money exchanged in this market will be either such that L-types are at least indifferent, which will make trade profitable for H-types, or such that H-types are indifferent, in which case trade will be profitable for L-types. Therefore, the natural arrangement in this market is actually always implemented in any equilibrium.} Let $p$ be the dollar price of one unit of investment. Conditional on being an L-type, and hence a seller, an agent who carries an amount $s$ solves the following problem\footnote{\,Given the competitive nature of this market, the only
variables that matter for the seller's problem are $s$ and the price (that she takes as given). Similarly, in the buyer's problem below the relevant variables are $m,p$.}
\begin{eqnarray}
\max_{q_s\leq s} & V^{3L}\left(m+pq_s,b,s-q_s \right)\nonumber
\end{eqnarray}or alternatively
\begin{eqnarray}
\max_{q_s\leq s} & \left\{\Lambda+\phi m+(\psi+R)b+(\psi+y_{_L})s+(\phi p-\psi-y_{_L})q_s \right\},\nonumber
\end{eqnarray}where $q_s$ is the supply of securities. The seller's optimal behavior yields the individual supply function
\begin{eqnarray}
q_s^*=\left\{\begin{array} {l@{\quad}l} 0,\, \textit{if}
\,\,\,p<\frac{\psi+y_{_L}}{\phi},
\\ \in[0,s], \textit{if} \,\,\, p=\frac{\psi+y_{_L}}{\phi},\\
s,\, \textit{if} \,\,\,p>\frac{\psi+y_{_L}}{\phi}.
\end{array}\right.\nonumber
\end{eqnarray} In short, L-types will not sell if the price is low, they will be indifferent if the price falls in a medium range, and would want to sell everything if the price exceeds the value of the asset. Interestingly, notice that since the price of $s$ is in nominal terms and that of the real asset is in real terms, inflation will have something to say here. Sellers compare $p$ to the nominal value of $\psi + y_{_L}$ by dividing it by $\phi$. Therefore, as the real value of money, $\phi$, increases it is more likely that $p>(\psi + y_{_L})/ \phi$ and they will want to sell everything. Conversely, as inflation is higher, $\phi \rightarrow 0$, sellers decrease their supply, $q_s \rightarrow 0$. \newline \indent Conditional on being an H-type, an agent solves
\begin{eqnarray}
\max_{q_b\leq m/p} & V^{3H}\left(m-pq_b,b,s+q_b \right)\nonumber
\end{eqnarray}or alternatively
\begin{eqnarray}
\max_{q_b\leq m/p} & \left\{\Lambda+\phi m+(\psi+R)b+(\psi+y_{_H})s+(\psi+y_{_H}-\phi p)q_b \right\}.\nonumber
\end{eqnarray}where $q_b$ is the demand of securities. The demand function is given by
\begin{eqnarray}
q_b^*=\left\{\begin{array} {l@{\quad}l} 0,\, \textit{if}
\,\,\,p>\frac{\psi+y_{_H}}{\phi},
\\ \in[0,m/p], \textit{if} \,\,\, p=\frac{\psi+y_{_H}}{\phi},\\
m/p,\, \textit{if} \,\,\,p<\frac{\psi+y_{_H}}{\phi}.
\end{array}\right.\nonumber
\end{eqnarray} The interpretation for buyers is the opposite as that for sellers, and inflation also plays a role; as $\phi \rightarrow 0$, buyers increase their demand, $q_b \rightarrow m/p$.
Figure 1 depicts the aggregate demand and supply curves. The equilibrium price and the quantity of $s$ that is traded depend on the shape of the demand curve. In particular, inspection of the demand and supply functions shows that equilibrium depends on the distribution of the returns on the financial asset. If $\phi m<(\psi +y_{_L})s$ (represented by $D_1$ in Figure 1), we have $Q^*=(\lambda/2)(\phi m)/(\psi+y_{_L})$. If $\phi m\geq (\psi +y_{_L})s$ (represented by $D_2$ in Figure 1), we have $Q^*=\lambda s/2$. Therefore, $Q^*=(\lambda/2)\min\left\{s,(\phi m)/(\psi+y_{_L})\right\}$. The parameter $\lambda$ is the probability of going into the DFM. It can be seen as the degree of access to decentralized financial markets. Thus, this parameter can also be interpreted as a measure of the liquidity of these markets. The equilibrium quantity of $s$ that a representative agent (H-type) acquires in the DFM is then
\begin{eqnarray}
q^*=\min\left\{s,\frac{\phi m}{\psi + y_{_L}} \right\},\label{equilq}
\end{eqnarray} and the equilibrium price is given by
\begin{eqnarray}
p^*=\left\{\begin{array} {l@{\quad}l} \frac{\psi + y_{_L}}{\phi},\, \textit{if}
\,\,\,\phi m<(\psi +y_{_L})s,
\\ \frac{m}{s}, \textit{if} \,\,\, \phi m\in [(\psi+y_{_L})s,(\psi +y_{_H})s],\\
\frac{\psi + y_{_H}}{\phi},\, \textit{if} \phi m>(\psi +y_{_H})s \,\,\,.
\end{array}\right.\label{equilp}
\end{eqnarray} Interestingly enough, the equilibrium in this market depends on real money balances, $\phi m$. Portfolio re-balancing generates a role for money and takes place at the equilibrium price, $p^*$, which depends on inflation through $\phi$. Moreover, it is key to notice from equation (\ref{equilq}) that trade in this market exists, $q^*>0$, if and only if $\phi m >0$. Therefore, the relationship between monetary equilibrium and financial trade is one-to-one.

\subsection{Investment Market}

At the beginning of every period, agents enter a financial market with holdings of money and real asset from previous period. Thus, the vector of state variables is $(m,a)$, and the agent wishes to maximize her continuation value in the DFM by choosing optimally what part of $a$ to invest in the risky asset. Therefore, an agent that enters the IM with portfolio $(m,a)$ has a value function\footnote{\,Since $b=a-s$ it does not
matter whether the agent chooses $s$ or $b$.}
\begin{eqnarray}
V^1(m,a)=\max_{s\in[0,a]}\left\{ \frac{1}{2}\left[ V^{2L}\left( m,b,s \right) +V^{2H}\left( m,b,s \right) \right] \right\}\label{V1}.
\end{eqnarray}

In order to examine the optimal choice of $s$ for the agent, we need to replace the value functions $V^{2j}$ with more useful expressions. To that end notice that
\begin{eqnarray}
V^{2L}\left( m,b,s \right)=\lambda V^{3L}\left( m+p q_s,b,s-q_s\right)+(1-\lambda)V^{3L}\left( m,b,s \right),\nonumber \\
V^{2H}\left( m,b,s \right)=\lambda V^{3H}\left( m-p q_s,b,s+q_s\right)+(1-\lambda)V^{3H}\left( m,b,s \right).\nonumber
\end{eqnarray}Equivalently, using the results derived above regarding the terms of trade in the DFM we can write
\begin{eqnarray}
V^{2L}\left( m,b,s \right)=\lambda V^{3L}\left( m+p^* q^*,b,s-q^*\right)+(1-\lambda)V^{3L}\left( m,b,s \right),\nonumber \\
V^{2H}\left( m,b,s \right)=\lambda V^{3H}\left( m-p^* q^*,b,s+q^*\right)+(1-\lambda)V^{3H}\left( m,b,s \right).\nonumber
\end{eqnarray}Then, exploiting the linearity of $V^{3j}$, equation (\ref{linearity}), and using (\ref{equilq}), we obtain
\begin{eqnarray}
\frac{1}{2}\left[ V^{2L}\left( m,b,s \right) +V^{2H}\left( m,b,s \right) \right]=\Lambda+\phi m+(\psi +R)a + \nonumber \\  +\left[ \frac{1}{2}\left(y_{_L}+y_{_H}\right)-R\right]s+\frac{\lambda}{2}\left(y_{_H}-y_{_L}\right)\min\left\{ s, \frac{\phi m}{\psi +y_{_L}} \right\}.\label{V11}
\end{eqnarray}
Inserting (\ref{V11}) into (\ref{V1}), we can rewrite the objective function as
\begin{eqnarray}
&\max_{s\in[0,a]}\Big\{ \left[ \frac{1}{2}\left(y_{_L}+y_{_H}\right)-R\right]s+\frac{\lambda}{2}\left(y_{_H}-y_{_L}\right)\min\left\{ s, \frac{\phi m}{\psi +y_{_L}} \right\} \Big\}\equiv &\nonumber \\
&\equiv \max_{s\in[0,a]}\Big\{ \alpha_1 s+\alpha_2 \min\left\{ s,\alpha_3 \right\} \Big\}.\nonumber &
\end{eqnarray}The definitions of the terms $\alpha_i$, $i=1,2,3$, are obvious and are adopted for analytical convenience. The objective function of the
agent is very intuitive. The first term, $\frac{1}{2}(y_{_L}+y_{_H})-R$, represents that for every unit of $a$ that she turns into $s$, she forgoes the return $R$, and gains $y_{_L}$ or $y_{_H}$ with equal probability. The second term
of the objective is the expected gain from trade in the DFM, which is equal to the total number of units of $s$ that trade in the DFM multiplied by
the average gain from this transaction. For every unit of $s$ that goes from the hands of an L-type to those of an H-type, a surplus equal to $y_{_H}-y_{_L}$ is generated.

\indent Consider now the optimal choice of $s$, namely, $s^*$. In order to focus on the more interesting situations we want to consider the case where the expected return from the risky investment is less than the return on the safe asset, i.e. we assume that $\alpha_1\equiv (1/2)\left(y_{_L}+y_{_H}\right)-R\leq0$. This means that if agents
were not able to trade $s$ in the DFM, they would never choose $s^*>0$.\footnote{\,Alternatively, the assumption that $\alpha_1\leq0$ shows the
robustness of the results in our model: if agents choose $s^*>0$ when the net return (excluding the potential gains in the DFM) of the risky asset is non-positive, then they will definitely do so if $\alpha_1$ were positive.} The objective of the agent is depicted in Figure 2 for various parameter values. The key variable for the determination
of $s^*$ is $\alpha_1+\alpha_2=(1/2)[(1+\lambda)y_{_H}+(1-\lambda)y_{_L}]-R$. This term is the multiplier of $s$ for $s\leq\alpha_3$. When the financial market is fairly liquid, $\lambda$ big, more weight is put on $y_{_H}$ and $\alpha_1+\alpha_2$ is also big. If financial markets are illiquid, $\lambda=0$, the agent never gets to trade into the DFM; the gain from holding $s$ coincides with the net
return of the risky asset, which is assumed to be non-positive. If $\lambda =1$, $\alpha_1+\alpha_2=y_{_H}-R>0$. The following Lemma summarizes the optimal investment policy.

\begin{lemma}
For a given state $(m,a)$, the optimal choice of $s\in[0,a]$ is given by
\begin{eqnarray}
s^*=\left\{\begin{array} {l@{\quad}l} 0,\, \textit{if} \,\,\,\alpha_1+\alpha_2<0, \\
\in\left[0,\min\{\alpha_3,a\}\right], \textit{if} \,\,\, \alpha_1+\alpha_2=0,\\
\min\{\alpha_3,a\}, \textit{if} \,\,\, \alpha_1+\alpha_2>0 \, \textit{and} \,\, \alpha_1<0,\\
\in\left[\min\{\alpha_3,a\},a\right], \textit{if} \,\,\, \alpha_1+\alpha_2>0 \, \textit{and}\,\, \alpha_1=0.
\end{array}\right.\label{optimals}
\end{eqnarray}
\label{lemmaoptimals}
\end{lemma}
\begin{proof}
The result follows from inspection of the objective function and Figure 2.
\end{proof} The following lemma will be useful in later analysis and expresses that the net gain of carrying assets across periods is non-positive. The interpretation behind this is that agents will only be willing to carry assets, money and real asset, if there are positive gains from re-balancing their investment positions in the DFM. Indeed, we show below that one of the main contributions of the model is that the value generated by this mechanism is reflected in the asset pricing equation. Should this gain not exist, only the fundamental value of the asset would be reflected by its price, i.e. the discounted flow of expected returns.
\begin{lemma}
In any equilibrium
\begin{eqnarray}
&\psi\geq \beta\left( R+\hat{\psi} \right),& \nonumber\\
&\phi\geq \beta \hat{\phi}.&\nonumber
\end{eqnarray}
\label{lemmanegativity}
\end{lemma}
\begin{proof}
This is a standard result in the literature and its proof is omitted. For details see for example Geromichalos, Licari, and Su\'{a}rez-Lled\'{o} (2007).
\end{proof}
Having characterized the optimal choice of $s$ we can proceed to analyze the choice of $\hat{m},\hat{a}$ by examining the value function
$V^{3j}$. This choice will crucially depend on the parameter values. Once again
the key determinant is the term $\alpha_1+\alpha_2$.
\newline\newline\indent \underline{\textbf{Case 1:} $\alpha_1+\alpha_2<0$}. We know that in the next period the agent will choose to not invest anything
in the risky asset, i.e. $\hat{s}^*=0$ and $q^*=0$. Hence, we can write\footnote{\,Since $\hat{s}^*=0$, we have $\hat{b}=\hat{a}$. Also, the money holdings in the next period of an agent that chooses $\hat{m}$ in the current period's CM are given by $\hat{m}+\mu M$.}
\begin{eqnarray}
V^{3j}(m,b,s)=U(X^*)-X^*+\phi m+(\psi+R)b+ (\psi+y_{_j})s\,\,\,\,\,\,\,\,\,\,\,\,\,\,\,\,\,\,\,\nonumber \\
+\max_{\hat{m},\hat{a}}\left\{-\phi\hat{m}-\psi\hat{a}+\frac{\beta}{2}\left[ V^{2L}(\hat{m}+\mu M,\hat{a},0)+V^{2H}(\hat{m}+\mu M,\hat{a},0) \right] \right\}=\nonumber\\
=\Omega_1+ \max_{\hat{m},\hat{a}}\left\{-\phi\hat{m}-\psi\hat{a}+\beta\left[ \hat{\Lambda} +\hat{\phi}(\hat{m}+\mu M)+(R+\hat{\psi})\hat{a} \right.\right.\,\,\,\,\,\,\,\,\,\,\,\,\,\,\,\,\,\,\,\,\nonumber\\ \left.\left.  +\frac{\lambda}{2}
\left(y_{_H}- y_{_L} \right)\min\left\{0,\frac{\hat{\phi}(\hat{m}+\mu M)}{\hat{\psi}+y_{_L}} \right\} \right] \right\}=\,\,\,\,\,\,\,\,\,\,\,\,\,\,\,\,\,\,\,\,\,\,\,\,\,\,\,\,\,\,\,\,\,\,\,\,\,\,\,\,\,\nonumber\\
=\Omega_2 + \max_{\hat{m},\hat{a}}\left\{-\left(\phi-\beta \hat{\phi} \right)\hat{m}-\left[\psi-\beta \left(R+\hat{\psi}\right) \right]\hat{a} \right\},\,\,\,\,\,\,\,\,\,\,\,\,\,\,\,\,\,\,\,\,\,\,\,\,\,\,\,\,\,\,\nonumber
\end{eqnarray}where the first equality follows from (\ref{V11}) and the definition of $q^*$. The terms $\Omega_1$, $\Omega_2$, and $\hat{\Lambda}$ do not depend on $\hat{m}, \hat{a}$, and their definitions are obvious.\footnote{\,In particular, $\hat{\Lambda}$ is the
term we get if we lead $\Lambda$, defined in (\ref{linearity}), by one period. Since $\Lambda$ is independent of $m,a$, $\hat{\Lambda}$ is independent of $\hat{m}, \hat{a}$.} Recall from Lemma \ref{lemmanegativity} that the multipliers of both $\hat{m}$ and $\hat{a}$ above are non positive. Hence, it is easy to see that the optimal choices of these variables satisfy
\begin{eqnarray}
\hat{m}^*=\left\{\begin{array} {l@{\quad}l} 0,\, \textit{if} \,\,\,\phi>\beta \hat{\phi},
\\ \in\Re_+, \textit{if} \,\,\, \phi=\beta \hat{\phi},
\end{array}\right.
\hat{a}^*=\left\{\begin{array} {l@{\quad}l} 0,\, \textit{if} \,\,\,\psi>\beta\left( R+\hat{\psi}\right),
\\ \in\Re_+, \textit{if} \,\,\, \psi=\beta\left( R+\hat{\psi}\right).
\end{array}\right.\label{optm}
\end{eqnarray}
We will return to describe equilibrium in this case. Before that, we complete our description of the optimal choice of $\hat{m},\hat{a}$. \newline\newline\indent \underline{\textbf{Case 2:} $\alpha_1+\alpha_2\geq0$}. Here $\hat{s}^*=\min\left\{\hat{a},\hat{\alpha}_3\right\}$, where $\hat{\alpha}_3=\hat{\phi}\left( \hat{m}+\mu M \right)/\left( \hat{\psi}+y_{{_L}} \right)$.\footnote{\,As opposed to $\alpha_3$, the terms $\alpha_1,\alpha_2$ are constant. Moreover, from Lemma \ref{lemmaoptimals} it follows that
$\hat{s}^*=\min\left\{\hat{a},\hat{\alpha}_3\right\}$ is always an optimal choice in the forthcoming period, although not the only one. This choice is the unique optimal,
only if $\alpha_1+\alpha_2>0$ and $\alpha_1<0$. Nevertheless, plugging $\hat{s}^*=\min\left\{\hat{a},\hat{\alpha}_3\right\}$ into the value function $V^{3j}$ is always
a good idea, since it yields the same result as any other $\hat{s}^*$.} Following the same strategy as above, we write
\begin{eqnarray}
V^{3j}(m,b,s)=U(X^*)-X^*+\phi m+(\psi+R)b+ (\psi+y_{_j})s\nonumber \\
+\max_{\hat{m},\hat{a}}\Big\{-\phi\hat{m}-\psi\hat{a}+\frac{\beta}{2}\left[ V^{2L}(\hat{m}+\mu M,\hat{a}-\min\left\{\hat{a},\hat{\alpha}_3\right\},\min\left\{ \hat{a},\hat{\alpha}_3\right\})\right.\nonumber\\\left.+V^{2H}(\hat{m}+\mu M,\hat{a}-\min\left\{\hat{a},\hat{\alpha}_3\right\},\min\left\{ \hat{a},\hat{\alpha}_3\right\}) \right] \Big\}=\nonumber\\
=\Omega_1+\max_{\hat{m},\hat{a}}\left\{ -\phi\hat{m} -\psi\hat{a} +\beta \left[ \hat{\Lambda} +\hat{\phi}\hat{m}+ \left( R+\hat{\psi} \right)\left( \hat{a}-\min\left\{ \hat{a},\hat{\alpha}_3\right\} \right) \right.\right. \nonumber\\ \left.\left. +\left( \hat{\psi}+\frac{y_{_H}+y_{_L}}{2} \right)
\min\left\{ \hat{a},\hat{\alpha}_3  \right\} + \frac{\lambda}{2}\left( y_{_H}-y_{_L} \right)\hat{q}^*\right] \right\}.\nonumber
\end{eqnarray}From (\ref{equilq}), the quantity of $s$ that changes hands in the DFM is given by
\begin{eqnarray}
\hat{q}^*=\min\left\{\hat{s}^*,\frac{\hat{\phi} \left( \hat{m}+\mu M \right)}{\hat{\psi} + y_{_L}} \right\}=\min\left\{\min\left\{\hat{a},\hat{\alpha}_3\right\},\frac{\hat{\phi} \left( \hat{m}+\mu M \right)}{\hat{\psi} + y_{_L}} \right\}=\nonumber\\
=\min\left\{\min\left\{\hat{a},\frac{\hat{\phi} \left( \hat{m}+\mu M \right)}{\hat{\psi} + y_{_L}}\right\},\frac{\hat{\phi} \left( \hat{m}+\mu M \right)}{\hat{\psi} + y_{_L}} \right\}=\min\left\{\hat{a},\frac{\hat{\phi} \left( \hat{m}+\mu M \right)}{\hat{\psi} + y_{_L}}\right\}.\nonumber
\end{eqnarray}Using this fact to rewrite the the value function $V^{3j}$, we conclude that the objective of the agent is to
\begin{eqnarray}
\max_{\hat{m},\hat{a}}\left\{-\left(\phi-\beta \hat{\phi} \right)\hat{m}-\left[\psi-\beta \left(R+\hat{\psi}\right) \right]\hat{a}\right.\,\,\,\,\,\,\,\,\,\,\,\,\,\,\,\,\,\,\,\,\,\nonumber\\
\left.+ \frac{\beta}{2} \left[ (1+\lambda)y_{_H}+(1-\lambda)y_{_L}-2R \right] \min\left\{\hat{a},\frac{\hat{\phi} \left( \hat{m}+\mu M \right)}{\hat{\psi} + y_{_L}}\right\} \right\}\equiv\nonumber\\
\equiv \max_{\hat{m},\hat{a}}\left\{-\gamma_1 \hat{m}-\gamma_2\hat{a}+ \gamma_3 \min\left\{\hat{a},\gamma_4 \hat{m}+\gamma_5\right\} \right\}.\nonumber\,\,\,\,\,\,\,\,\,\,\,\,\,\,\,\,\,\,\,
\end{eqnarray}

\indent The definitions of the terms $\gamma_i$, $i=1,..,5$ is obvious and greatly simplifies the notation. Notice that $\gamma_1,\gamma_2\geq0$. Also, $\gamma_3=
\beta(\alpha_1+\alpha_2)\geq0$. Finally, $\gamma_4>0$ and the sign of $\gamma_5$ depends on whether the monetary authority is running inflation or deflation,
i.e. the sign of $\mu$. The objective function is linear in $\hat{m},\hat{a}$. The optimal solution depends on the magnitude of the various
gamma's. Considering that both arguments in the $min$ function depend on decision variables of the agent, the optimal choice will have both arguments being equal to each other. Thus, it is relatively easy to verify that\footnote{\,In the third line, the condition $z\geq-\gamma_5/\gamma_4$ just guarantees that if $\mu<0$ the optimal choice
of $\hat{a}$ satisfies non-negativity.}
\begin{eqnarray}
\left(\hat{m}^*,\hat{a}^*\right)=\left\{\begin{array} {l@{\quad}l}
(0,0),\, \textit{if} \,\,\,\gamma_1+\gamma_2\gamma_4>\gamma_3\gamma_4, \\
(+\infty,+\infty), \, \textit{if} \,\,\, \gamma_1+\gamma_2\gamma_4<\gamma_3\gamma_4,\\
(z,\gamma_4 z+\gamma_5), \, \textit{for any} \,\, z\geq-\frac{\gamma_5}{\gamma_4},\, \textit{if} \,\,\, \gamma_1+\gamma_2\gamma_4=\gamma_3\gamma_4.
\end{array}\right.\label{optimalAM}
\end{eqnarray}

\section{Equilibrium}
We now proceed to the definition and the analysis of equilibrium.
\begin{definition}
An equilibrium for this economy is a set of value functions $V^{ij}$, $i=1,2,3$ and $j=L,H$ that satisfy the Bellman equations, a triplet $\left( p^*,q^*,s^* \right)$ that satisfy (\ref{equilq}), (\ref{equilp}), and (\ref{optimals}) in every period, and a pair of bounded sequences $\{\phi_t M_t\}_{t=0}^{\infty}$,$\{ \psi_t \}_{t=0}^{\infty}$, such that the agent behaves optimally under the conditions $a_t=A$ and $m_t=M_t$, all $t$.
\end{definition}
Consider first Case 1, i.e. $\alpha_1+\alpha_2<0$. From (\ref{optm}), a necessary condition for equilibrium is $\psi=\beta\left( R+\hat{\psi}\right)$, which implies that
the sequence of the asset price should follow the difference equation $\psi_{t+1}=-R+(1/\beta)\psi_t$. Since $\beta<1$, this can only be true if
\begin{eqnarray}
\psi_t=\bar{\psi}=\frac{\beta R}{1-\beta}. \label{fundamental}
\end{eqnarray}The asset price given by (\ref{fundamental}) is the fundamental value of the asset, i.e. the discounted stream of future dividends. The result according to which $\psi_t=\bar{\psi}$, reflects the fact that in the case under consideration the asset is only valued for the fruit it yields: agents never buy it in order to invest in the risky asset and possibly trade in the DFM.
\newline\indent The role of money in this economy is not essential. Agents have a positive demand for money only if $\phi=\beta \hat{\phi}$. In steady state this
implies that $\mu=1-\beta$, i.e the monetary authority is following the Friedman rule.\footnote{\,In steady state $\phi M=\hat{\phi}\hat{M}\Rightarrow \phi M=\hat{\phi}(1+\mu)M$. Therefore, $1+\mu=\phi/\hat{\phi}=\beta$.} Agents carry money only when the cost of doing so is zero (equivalently, when the nominal interest rate is zero). Of course, they never get to use that
money in the DFM, since $s^*=0$ and there is nothing to trade money with.
\newline\indent The more interesting case is therefore $\alpha_1+\alpha_2\geq0$. For the rest of this paper we assume that this condition holds. Optimality (described by (\ref{optimalAM})) reveals that a necessary condition for the existence of equilibrium is
\begin{eqnarray}
&\gamma_1+\gamma_2\gamma_4=\gamma_3\gamma_4 \Leftrightarrow& \nonumber \\
&\phi-\beta \hat{\phi}+\left[ \psi-\beta\left( R+\hat{\psi}\right) \right] \frac{\hat{\phi}}{\hat{\psi}+y_{_L}}=\beta(\alpha_1+\alpha_2)\frac{\hat{\phi}}{\hat{\psi}+y_{_L}}.&\label{diff}
\end{eqnarray} As we pointed in the definition above, we focus on steady state equilibria. That is, equilibria in which money balances and asset prices are constant in every period. However, in the appendix we show that bubbles, i.e. equilibria in which the asset price is above the market price for some periods, can never arise in this environment.
\newline\indent To simplify equation (\ref{diff}), divide both sides by $\hat{\phi}$ and recall that at the steady state $\phi/\hat{\phi}=1+\mu$. We obtain
\begin{eqnarray}
\left( 1+\mu-\beta \right)\left( \hat{\psi}+y_{_L} \right)= -\psi+\beta\left( R+\hat{\psi}\right) +\beta (\alpha_1+\alpha_2).\nonumber
\end{eqnarray} Using the definition of $\alpha_1+\alpha_2$ and solving with respect to $\hat{\psi}$ yields the asset pricing difference equation followed,
\begin{eqnarray}
\hat{\psi}=\frac{\frac{\beta}{2}\left[ (1+\lambda)y_{_H}+(1-\lambda)y_{_L} \right]-(1+\mu-\beta)y_{_L}}{1+\mu-2\beta}-\frac{1}{1+\mu-2\beta}\psi.\label{differ}
\end{eqnarray} As we explained, we drive our attention to the case in which $\psi=\hat{\psi}=\psi^*$.\footnote{\,The returns $y_{_H}$ and $y_{_L}$ are deterministic and with probability $0.5$. Therefore the only randomness comes from the idiosyncratic realizations of those returns. The probability of going to the DFM, $\lambda$, is also constant. Therefore, the intuition behind this is that the gain from decentralized trade is also always going to be the same, so it makes sense that the price of the asset is constant every period as well. More extensive discussion is provided later.} Using this in (\ref{differ}) and solving with respect to $\psi^*$ we obtain \begin{eqnarray}
\psi^*=\frac{\frac{\beta}{2}\left[ (1+\lambda)y_{_H}+(1-\lambda)y_{_L} \right]-(1+\mu-\beta)y_{_L}}{2(1-\beta)+\mu}.\label{psi}
\end{eqnarray}
\noindent It immediately follows from this equation that the price of the asset depends negatively on the growth rate of money. This result is
in contrast with the predictions of related papers, like Geromichalos et al (2007) or Lester, Postlewaite, and Wright (2007). In these papers agents bring money and real assets with them in the DFM in order
to buy a specialized good. Higher inflation makes carrying money more costly, and so people turn to the relatively cheaper asset, thus increasing its price. Unlike these papers, here money and the asset are complements, in the sense that agents need both objects in order to trade in the DFM.\footnote{\,More precisely, agents trade $m$ for $s$ in the DFM, but in order to be in possession of some $s$, agents need to bring some
$a$ from the previous period. To see why $m$ and $s$ are complements consider the extreme case in which all agents bring a billion dollars but $s^*=0$. The surplus generated in the DFM is still zero.} If money is costly to carry, the demand for the asset will decrease and so will its equilibrium price (given fixed supply).
\newline\indent A second observation that follows from (\ref{psi}) is that the price of the asset is always higher than the fundamental value given by (\ref{fundamental}). In the Appendix we show that this result is true as long as $\lambda>0$. Moreover, since the asset price is decreasing in the money growth rate, it obtains its maximum value when $\mu=\beta-1$. The resulting expression is given by
\begin{eqnarray}
\psi_{_{max}}=\frac{\beta}{1-\beta}\frac{1}{2}\left[ (1+\lambda)y_{_H}+(1-\lambda)y_{_L} \right]\geq \frac{\beta}{1-\beta}R, \nonumber
\end{eqnarray}where the inequality follows from the fact that $\alpha_1+\alpha_2\geq0$. Again, unless $\lambda =0$, the asset price is strictly greater than its fundamental value. Furthermore, since the price, $\psi^*$, can never be smaller than $\beta R/(1-\beta)$ because this would violate Lemma \ref{lemmanegativity}, the upper bound of admissible monetary policies can be uniquely defined by $\bar{\mu}\equiv\left\{ \mu: \psi^*(\mu)=\beta R/(1-\beta) \right\}$. A characterization of $\bar{\mu}$ is provided in the Appendix where we also show that equilibria with positive inflation, $\bar{\mu}>0$, can be supported depending on the distribution of idiosyncratic shocks. \newline \indent The asset pricing equation in this model reflects two dimensions. First, the asset is valued for the dividend it yields, $R$. Second the asset is also valued for its property to be transformed into $s$, thus allowing agents to generate additional value by re-balancing their portfolios in the DFM. Since access to the DFM is granted with probability $\lambda$, the difference expressed by $\psi^*-\beta R/(1-\beta)$ reflects the premium of the asset over the fundamental value. In particular, this premium is largest at the Friedman rule, in which case the term $\alpha_1+\alpha_2\equiv (1/2)\left[ (1+\lambda)y_{_H}+(1-\lambda)y_{_L} \right]-R$ would represent the spread. Notice, however, that in general such premium depends on the money growth rate, $\mu$. Then, we consider it appropriate to refer to it as a \textit{liquidity premium}. After all, additional value in this economy is generated by the degree of \symbol{92}liquidity\symbol{34} in the DFM. That is, value is generated along with the easiness to reallocate portfolios in the DFM, and this ultimately depends on the value of money.

\subsection{Welfare} In most other related models welfare is measured in a standard manner by two terms: the excess utility over the cost of producing the general consumption good in the centralized market, $U(X)-X$, and the excess utility of the total amount of consumption in the decentralized market over the cost of producing that total amount. However, in our economy there is no consumption good being traded in the decentralized market. Instead, the amount of assets traded in the DFM affects the budget constraint in the CM. Therefore, the way we measure welfare in this economy is by computing the total value of trade in the DFM in terms of the CM. \newline \indent Given that $\alpha_1+\alpha_2>0$, in the IM agents place all their asset holdings into the risky investment, $s^*=A$. In the DFM the amount of assets that change hands is given by $Q^*=(\lambda/2)\min\left\{ A,\phi M/(\psi+y_{_L}) \right\}=(\lambda/2)A$.\footnote{\,Remember that in any equilibrium with positive and finite money and asset holdings, these variables are chosen so that the two arguments in the $min$ function are the same.} This means that L-types sell all their assets to H-types, which is efficient from a social point of view, and in return they receive $p^*Q^*=\lambda A(\psi^*+y_{_L})/(2\phi^*)$. The latter accounts for the total value of trade in the DFM. One result that is as stricking as important is that total welfare is unaffected by inflation as long as $\mu\in\left[\beta-1,\bar{\mu}\right]$. This result implies, in contrast to other monetary models, that welfare gains can be sustained even by positive rates of inflation.\footnote{\,This result is closely connected to the nature of the returns to the risky asset. In an alternative specification of this model, we consider the case in which investing $s\in[0,a]$ units of the asset yields $\epsilon_j f(s)$, where $j=L,H$ and $f$ is a strictly increasing and concave function with standard Inada conditions. Because $lim_{s\rightarrow 0}f'(s)=\infty$, it is never optimal for the L-type to give up all her assets. This creates a link between inflation and the amount of $s$ that changes hands in the DFM. However, the analysis becomes very complex and we are only able to solve the model numerically. Since none of the major results of the model are altered, we prefer the more simple specification presented here, which yields closed form solutions and clear intuition.} If $\mu>\bar{\mu}$ carrying money becomes too expensive and the monetary equilibrium collapses, together with any trade in the DFM. A first best is achieved when $\lambda=1$ because decentralized trade (and, therefore, optimal re-alocation of resources) is maximized.
The following proposition summarizes the most important results.
\begin{proposition}
The key variable for the determination of equilibrium is $\alpha_1+\alpha_2$, which can be interpreted as the net return of the risky asset.
\newline\noindent \textbf{(a)} If $\alpha_1+\alpha_2\leq0$ no trade exists in the DFM, $q^*=s^*=0$, and the asset is only valued for the dividend it yields, i.e. $\psi^*=\beta R/(1-\beta)$. There is no essential role for money in this economy. Its real allocation coincides with that of an Arrow-Debreu economy and asset pricing follows the Lucas formula. A monetary equilibrium could only be supported by the Friedman rule.
\newline\noindent \textbf{(b)} If  $\alpha_1+\alpha_2>0$, $\psi_t=\psi^*(\mu)$ given by (\ref{psi}) for all $t$. The range of policies under which monetary equilibria can be supported is $\left[\beta-1,\bar{\mu} \right]$, where $\bar{\mu}$ is defined by $\bar{\mu}\equiv\left\{ \mu: \psi^*(\mu)=\beta R/(1-\beta) \right\}$. For all $\mu$ in this range, $d\psi^*/d\mu<0$. The bigger the $\lambda$, the bigger the premium between the actual value and the fundamental value of the asset. This liquidity premium is maximized at the Friedman rule. These results are presented in Figure 3. \newline\indent Equilibrium welfare depends on $\lambda$, i.e. the probability with which agents get to trade in the DFM and re-balance their positions. In the DFM L-types sell all their risky assets to H-types, hence all equilibria are (constrained) efficient. The First Best is achieved when $\lambda =1$.
\end{proposition} \noindent If the decentralized financial market can generate value ($\alpha_1+\alpha_2>0$), then monetary equilibria, $\phi m >0$, exist and can be supported by monetary policies ranging from the Friedman rule to potentially positive inflation rates. Risky investments take place even if their expected return is worse than that of safe investments. Agents access the financial market to re-balance their portfolios so that investors with low returns sell transfer all their security holdings, $q^*=s^*=A$, in exchange for $p^*$ units of money. \newline \indent Thus, money plays an essential role as a medium of exchange in frictional financial markets. In particular, the price at which securities are traded, $p^*$, not only is related to the price of the real asset, but will also depend negatively on the money growth rate, $\mu$. The quantity exchanged, $s^*$, is independent of monetary considerations. More interestingly, the price of the real asset, $\psi^*=\psi(\mu)$, is always above its fundamental value and exhibits a liquidity premium that depends crucially on two notions of liquidity. On one side, the degree of access to decentralized financial markets, $\lambda$. And on the other side, the degree of liquidity within the financial markets, which depends on money supply, $\mu$. The liquidity premium of the real asset, and therefore its value, is maximized at the lowest feasible inflation rate. Positive money growth rate could also generate welfare gains.

\section{Conclusion}

Certain types of markets have quickly developed to attain an enormous size within the financial sector. They have also recently proven to have a considerable impact. Their structure deviates from the standard centralization of primary markets that is embedded in most macroeconomic models. Markets for securitization and Over-the-Counter markets stand as a reference of these type of markets. In this context a concern arises naturally: does monetary policy have something to say in these environments? \newline \indent We have presented a general equilibrium model that analyzes the role of money in decentralized financial markets. We provided an innovative framework that models a risky financial asset that can be issued from the value of a safe-return real asset. Agents can re-balance their positions on these relatively illiquid assets with a more liquid instrument (fiat money) in a decentralized financial market. We show that value can be generated from this transfer of assets, and therefore welfare can be increased, by reallocating otherwise inefficient investment. \newline \indent Our new approach yields very interesting results, some in contrast with the existing literature. On the one hand, monetary policy affects the price of the real and the financial asset, attaining their maximum values at a policy equivalent to the Friedman rule. In particular, since money and the real asset are complements, its price is negatively correlated with inflation. On the other hand, monetary policy then influences the value generated for the economy in the decentralized market. In fact, policy can always sustain asset prices above their fundamental values and generate welfare gains even under inflationary money growth rates, as long as those rates are not too large.

\begin{appendix}

\section{Appendix}

\textbf{Non-steady state equilibria.}
\newline\noindent We prove here that equilibria other than the steady state described above, cannot arise in this model. More precisely, we cannot have sequences $\{\phi_t M_t\}_{t=0}^{\infty}$,$\{ \psi_t \}_{t=0}^{\infty}$ that converge to a certain limit either monotonically or in an oscillating manner. Therefore, the only way to keep these sequences bounded is to have $\phi_t M_t$ and $\psi_t$ constant in all periods (in order to make this statement one needs to exclude cycles, which is the case here). As part of an equilibrium, we are looking for bounded sequences of money balances and asset prices such that (\ref{diff}) holds and also
\begin{eqnarray}
\phi_t M_t=(\psi_t+y_{_L})A,\,\,\, \forall t.\label{A1}
\end{eqnarray} Define real balances as $z\equiv\phi M$. If we multiply (\ref{diff}) by $\hat{M}$ and use (\ref{A1}) to substitute for $\hat{\phi}/(\hat{\psi}+y_{_L})$ and $\psi$, we conclude that real money balances follow a first-order linear difference equation,
\begin{eqnarray}
\hat{z}=-\frac{ \beta\left(\alpha_1+\alpha_2+R\right)+(1-\beta)y_{_L}}{2\beta}A+\frac{2+\mu}{2\beta}z.\nonumber
\end{eqnarray}However, $\mu\geq\beta-1$ always, which implies for the multiplier of $z$ that $(2+\mu)/(2\beta)>1$. Therefore, the sequence $\{\phi_t M_t\}_{t=0}^{\infty}$ will always be explosive, unless $(\phi M)_t=(\phi M)^*$ for all $t$. Since (\ref{A1}) has to hold in every period the same conclusion is true for $\{\psi_t\}_{t=0}^{\infty}$.

\noindent\textbf{Optimal monetary policy range.}
\newline\noindent In the section on equilibrium with the price taking version we discussed the feasible range for optimal monetary policies. Here we provide a more detailed characterization of the upper bound of that range. The whole point of the price taking version is that it supports asset prices higher than the fundamental value, and such price is decreasing in the money growth rate. Therefore, the upper bound for monetary policy is the money growth rate that makes the price of the real asset equal to its fundamental value, $\bar{\mu}\equiv \{\mu: \psi^*(\mu)=\beta R/(1-\beta)\}$. If we take the equation for the equilibrium price of the asset, (\ref{psi}), and make it equal to its fundamental value, we can solve for the money growth rate, $\bar{\mu}_b$. Thus, we write \begin{eqnarray} \psi^*=\frac{\frac{\beta}{2}[(1+\lambda)y_{_H}+(1-\lambda)y_{_L}]-(1-\mu -\beta)y_{_L}}{2(1-\beta)+\mu}=\frac{\beta R}{1-\beta}. \nonumber \end{eqnarray} After some algebra one arrives at the following expression \begin{eqnarray} \bar{\mu}_b=\beta -1+\frac{\beta (1-\beta)(\alpha_{_1}+\alpha_{_2})}{y_{_L}(1-\beta)+\beta R}. \nonumber \end{eqnarray} which will be positive as long as $\alpha_{_1}+\alpha_{_2}>\frac{1-\beta}{\beta}y_{_L}+R>0$.

\noindent\textbf{Bargaining in the DFM.}
\newline\noindent We present here the complete analysis of the version of the model with bargaining in the DFM. We consider agents bilaterally trading in this market according to a Nash bargaining procedure. It can be easily shown that in any meeting between two agents of the same type no trade will occur. Thus, we focus here on bilateral meetings of agents of different type. In this matches the portfolio of a low type agent is $(m,b,s)$ and that of a high type agents is $(\bar{m},\bar{b},\bar{s})$. Thus, agents must choose the amount of money, $d_m$, and securities, $d_s$, exchanged in order to \begin{eqnarray}\max_{d_m,d_s}& \left[V^{3L}(m+d_m,b,s-d_s)-V^{3L}(m,b,s)\right]^{\theta}\times\nonumber\\&
\times\left[V^{3H}(\bar{m}-d_m,\bar{b},\bar{s}+d_s)-V^{3H}(\bar{m},\bar{b},\bar{s})\right]^{1-\theta}\nonumber\\
\textrm{s.t.}&-m \leq d_m\leq \bar{m}\textrm{; }-\bar{s} \leq d_s\leq s\nonumber\end{eqnarray} In other words, agents maximize the value of trade. In principle, we formulate the problem in the most general way that allows either agent to sell or buy securities. However, in the equilibrium delivered by this version, in which returns on the risky asset are linear in the amount of the asset, H-types will always be the buyers and L-types the sellers. Now, using the linearity of the value function, we can rewrite the bargaining problem as \begin{eqnarray}\max_{d_m,d_s}& \left[\phi d_m-(\psi+y_{_L})d_s\right]^{\theta}\left[-\phi d_m + (\psi+y_{_H})d_s)\right]^{1-\theta}\label{linear bargaining}\\
\textrm{s.t.}&-m \leq d_m\leq \bar{m}\textrm{; }-\bar{s} \leq d_s\leq s\nonumber\end{eqnarray} Before proceeding with any further analysis, it is absolutely critical to realize that the problem we pose here is completely different from most other bargaining problems in the mainstream literature on money search, at least in those models derived from LW. The key aspect is the following. In those papers the welfare pie in the DFM is given by $u(q)-c(q)$, where $q$ is the consumption of the good in that market, and $u(q)$, $c(q)$ are the utility and cost derived from it. Whereas in our model the size of the pie is not pre-determined. In fact, it depends on the size of the portfolio that the parties carry with them. Therefore, the problem we propose here can only be solved by imposing constraints, otherwise the agents would optimally choose $d_m=d_s=+\infty$.\footnote{\,In all search models where consumption is traded in the DFM, its maximum level is bounded above by the condition $u'(q^*)=c'(q^*)$.} Thus, the optimal solution to our problem features one of the constraints above binding in equilibrium. We will proceed case by case and then summarize the results. \newline \newline \indent Suppose first that $d_m=\bar{m}$ and $d_s \leq s$. The problem has a unique interior solution determined by the FOC for $d_s$: \begin{eqnarray} \frac{\theta(\psi+y_{_L})}{\phi \bar{m}-(\psi+y_{_L})d_s}=\frac{(1-\theta)(\psi+y_{_H})}{-\phi \bar{m}+(\psi+y_{_H})d_s}, \nonumber \end{eqnarray} which yields \begin{eqnarray} d_s=\frac{\phi \bar{m}[\psi+(1-\theta)y_{_H}+\theta y_{_L}]}{(\psi+y_{_L})(\psi+y_{_H})}\equiv d_s^*. \nonumber \end{eqnarray} To be more precise, $d_s=\min\{s,d_s^*\}$. Now let us suppose that $d_s=s$ and $d_m\leq \bar{m}$. The first order condition for $d_m$ can be written as \begin{eqnarray} \theta [-\phi d_m+(\psi+y_{_H})s]=(1-\theta)[\phi d_m-(\psi+y_{_L})s], \nonumber \end{eqnarray} which implies that \begin{eqnarray} d_m=\frac{s[\psi+\theta y_{_H}+(1-\theta) y_{_L}]}{\phi}\equiv d_m^*. \nonumber \end{eqnarray} So, in fact, $d_m=\min\{\bar{m},d_m^*\}$. In more detail, since there can be equilibria where only one of the constraint binds or both bind, we have that if $\bar{m}\leq d_m^*$ then $d_m=\bar{m}$. In this case, $d_s=s$ only if $s \leq d_s^*$. In other words, both constraints would be binding $(d_m=\bar{m}, d_s=s)$ in an equilibrium where $\bar{m}\leq \frac{s}{\phi}[\psi+\theta y_{_H}+(1-\theta)y_{_L}]\equiv m_h$ and $\bar{m}\geq \frac{s}{\phi}\frac{(\psi+y_{_L})(\psi+y_{_H})}{[\psi+(1-\theta) y_{_H}+\theta y_{_L}]}\equiv m_l$. It is easy to check that \begin{eqnarray} \psi+\theta y_{_H}+(1-\theta)y_{_L}\geq \frac{(\psi+y_{_L})(\psi+y_{_H})}{[\psi+(1-\theta) y_{_H}+\theta y_{_L}]}, \nonumber \end{eqnarray} and this condition holds with strict inequality unless $y_{_H}=y_{_L}$, or $\theta =0$, or $\theta =1$; but none of these happens in our model. Therefore, for all $\bar{m} \in [m_l,m_h]$ agents exchange $d_m=\bar{m}$ and $d_s=s$. The cases where only one constraint binds are described by the pairs $(d_m=\bar{m}<m_l,d_s=d_s^*<s)$ and $(m_h=d_m=d_m^*<\bar{m},d_s=s)$. That is, if money holdings of the buyer are enough to purchase all $s$ from the seller, then $d^*_m$ is handed over in exchange for all $s$. Otherwise, the buyer will trade away all her money in exchange for $d^*_s$. All these possible cases are described by the solution below \begin{eqnarray} d_m=\min\left\{\bar{m},\frac{s[\psi+\theta y_{_H}+(1-\theta) y_{_L}]}{\phi}\right\},\nonumber\\\label{linear bargaining solution} \\d_s=\min\left\{s,\frac{\phi \bar{m}[\psi+(1-\theta)y_{_H}+\theta y_{_L}]}{(\psi+y_{_H})(\psi+y_{_L})}\right\}.\nonumber\end{eqnarray} Once we know how the equilibrium in the DFM looks like, we can solve the rest of the model proceeding as in the previous section. We use the solution to the DFM in the optimization of the first market, and then we find the optimal solution for $\hat{m}$ and $\hat{a}$. We first have to find $s$ so that \begin{eqnarray}
V^1(m,a)=\max_{s\in[0,a]}\left\{ \frac{1}{2}\left[ V^{2L}\left( m,b,s \right) +V^{2H}\left( m,b,s \right) \right] \right\}\label{V1 for s barg}.
\end{eqnarray} where now \begin{eqnarray}V^{2L}\left( m,b,s \right)\nonumber=\frac{\lambda}{2}V^{3L}\left( m+\int d_m(\tilde{m},s)dF(\tilde{m}),b,s-\int d_s(\tilde{m},s)dF(\tilde{m}) \right)&&\nonumber\\+\left(1-\frac{\lambda}{2}\right)V^{3L}\left( m,b,s \right),&&\label{V2L barg}\end{eqnarray} and \begin{eqnarray}V^{2H}\left( m,b,s \right)\nonumber=\frac{\lambda}{2}V^{3H}\left( m-\int d_m(m,\tilde{s})dG(\tilde{s}),b,s+\int d_s(m,\tilde{s})dG(\tilde{s}) \right)&&\nonumber\\+\left(1-\frac{\lambda}{2}\right)V^{3H}\left( m,b,s \right).&&\label{V2H barg}\end{eqnarray} In the previous value functions, $F(\tilde{m})$ and $G(\tilde{s})$ are the distribution of money and security holdings in the economy. Using the solution to the bargaining problem, (\ref{V1 for s barg}) can be written as \begin{eqnarray} &V^1(m,a)=\max_{s \in [0,a]} \left\{\Lambda +\phi m +(\psi +R)a+\left(\frac{1}{2}y_{_L}+\frac{1}{2}y_{_H}-R\right)s\right. \nonumber \\ &\left. + \frac{\lambda \phi}{4}\min\left\{\bar{m},\frac{[\psi+\theta y_{_H}+(1-\theta)y_{_L}]}{\phi}s\right\}- \frac{\lambda \phi}{4}\min\left\{m,\frac{[\psi+\theta y_{_H}+(1-\theta)y_{_L}]}{\phi}\bar{s}\right\}\right. \nonumber \\&\left. + \frac{\lambda (\psi+y_{_H})}{4}\min\left\{\bar{s},\frac{\phi[\psi+(1-\theta) y_{_H}+\theta y_{_L}]}{\phi}m\right\}- \frac{\lambda (\psi+y_{_L})}{4}\min\left\{s,\frac{\phi[\psi+(1-\theta) y_{_H}+\theta y_{_L}]}{\phi}\bar{m}\right\}\right\}. \nonumber \end{eqnarray} Remember that at this stage we have to choose $s$, so we are only interested in the 4th, 5th, and 8th terms. Therefore, focusing only on those terms and rearranging them a little bit we can write the relevant objective function as \begin{eqnarray} \max_{s \in [0,a]} \left\{\left(\frac{1}{2}y_{_L}+\frac{1}{2}y_{_H}-R\right)s + \frac{\lambda}{4}\left[[\psi+\theta y_{_H}+(1-\theta)y_{_L}]\min\left\{\frac{\phi \bar{m}}{\psi+\theta y_{_H}+(1-\theta)y_{_L}},s\right\}\right.\right. &&\nonumber \\\left.\left. - (\psi+y_{_L})\min\left\{\frac{\phi[\psi+(1-\theta) y_{_H}+\theta y_{_L}]}{(\psi+y_{_H})(\psi+y_{_L})}\bar{m},s\right\}\right]\right\}.&& \nonumber \end{eqnarray} or in short \begin{eqnarray} \max_{s \in [0,a]} \left\{\left(\frac{1}{2}y_{_L}+\frac{1}{2}y_{_H}-R\right)s + \frac{\lambda}{4}\left[\alpha_1^b\min\left\{\alpha_2^b,s\right\} - \alpha_3^b\min\left\{\alpha_4^b,s\right\}\right]\right\}, \nonumber \end{eqnarray} where the definitions of $\alpha_i^b, i=1,...,4,$ are obvious. It is easy to check that $\alpha_i^b \geq 0,\alpha_1^b,\alpha_2^b>0,$ and most importantly, $\alpha_1^b\geq \alpha_3^b,\alpha_4^b>\alpha_2^b$. Now, let us recall that we are focusing on the interesting case $(1/2)(y_{_H}+y_{_L})-R\leq 0$. The following can easily be verified by visual inspection after straightforward substitution. If we choose $s \leq \alpha_2^b$ the objective function above is strictly increasing in $s$ as long as $(\lambda \theta)/4(y_{_H}-y_{_L})>-[(1/2)(y_{_H}+y_{_L})-R]$. On the other hand, if our choice was $s \in [\alpha_2^b,\alpha_4^b]$ the objective would become $(\lambda /4)\phi \bar{m}+[(1/2)(y_{_H}+y_{_L})-R-(\lambda /4)(\psi+y_{_L})]s$, which is clearly decresing in $s$. Finally, if $s \geq \alpha_4^b$ the objective would just be decreasing in $s$ with all but the first term being constant. In a word, the objective function is decreasing for all $s$ beyond $\alpha_2^b$, which yields the following characterization of the solution \begin{eqnarray}
s^*=\left\{\begin{array} {l@{\quad}l} 0,\, \textit{if} \,\,\,\frac{1}{2}(y_{_H} + y_{_L})< R-\frac{\lambda \theta (y_{_H}-y_{_L})}{4}, \\
\in\left[0,\min\{\alpha^b_2,a\}\right], \textit{if} \,\,\,\frac{1}{2}(y_{_H} + y_{_L}) = R-\frac{\lambda \theta (y_{_H}-y_{_L})}{4},\\
\min\{\alpha^b_2,a\}, \textit{if} \,\,\, \frac{1}{2}(y_{_H} + y_{_L})> R-\frac{\lambda \theta (y_{_H}-y_{_L})}{4}.
\end{array}\right.\label{optimal s bargaining}
\end{eqnarray} where $\alpha^b_2=\frac{\phi \bar{m}}{\psi+\theta y_{_H}+(1-\theta)y_{_L}}$. As this solution describes, that the DFM actually opens, $s^*>0$, depends on whether the expected return on the risky investment plus the value generated by re-balancing portfolios with bargaining is high enough. If the expected return of securities is low, $\frac{1}{2}(y_{_H} + y_{_L})\leq R-\frac{\lambda \theta (y_{_H}-y_{_L})}{4}$, it is optimal to set $s^*=0$ and no trade happens in the DFM. All we have to do then is solve for $\hat{m}$ and $\hat{a}$. In this case people store the safe asset only if its price is constant and behaves according to $\psi=\beta (R+\hat{\psi})$, that is $\psi = \beta R/(1-\beta)$. Obviously, in this equilibria money is not essential and will only circulate as long as $\phi = \beta \hat{\phi}$. On the other hand, if the expected return on securities is good enough, $\frac{1}{2}(y_{_H} + y_{_L})\geq R-\frac{\lambda \theta (y_{_H}-y_{_L})}{4}$, then $s^*=\min \{a,\alpha^b_2\}$ and we would proceed to solve for $\hat{m}$ and $\hat{a}$ using this choice of $s$. Surprisingly enough, even though in this case there could be trade in the DFM, the only equilibrium that can be supported again features the fundamental value of the real asset, $\psi = \beta R/(1-\beta)$. \newline \newline \indent Let us consider first the case of a bad distribution of returns, i.e., $(1/2)(y_{_H}+y_{_L})\leq R-(\lambda \theta /4)(y_{_H}-y_{_L})$. As we said, $s^*=0$ and we have to choose $\hat{m}$ and $\hat{a}$ to solve \begin{eqnarray} &&V^{3j}(m,b,s)= \nonumber \\ &&\kappa(m,b,s)+\max_{\hat{m},\hat{a}} \left\{-\phi \hat{m}-\psi\hat{a}+\beta \left[\frac{1}{2}V^{2L}(\hat{m}+\mu M,\hat{a},0)+\frac{1}{2}V^{2H}(\hat{m}+\mu M,\hat{a},0)\right]\right\}, \nonumber \end{eqnarray} for $j=H,L$, where $\kappa(m,b,s)$ is just a group of terms depending on previous choices of money, real assets, and securities, but does not depend on the actual choice variable. Also, remember that in the next period an amount $\mu M$ of money is injected (subtracted) in (from) the economy. The details of the rest of the computation are available upon request. All that matters is that the optimal decision for next period real asset holdings, $\hat{a}$, is given by \begin{eqnarray}\hat{a}^*=\left\{\begin{array}{lr}0,&\textrm{if }\psi>\beta (\hat{\psi}+R), \\\in \mathbb{R}_+,&\textrm{if }\psi=\beta (\hat{\psi}+R), \end{array}\right.\label{optimal ahat bad appendix}\end{eqnarray} and the optimal next period money holdings, $\hat{m}$, will only be positive under certain parameterizations, when $\phi=\beta \hat{\phi}$. In brief, if the distribution of returns on securities is bad, people will optimally choose not to issue any risky assets. As a consequence, there will be no trade in the DFM and money will only circulate under the Friedman rule. The only feasible path for the asset price is then $\psi =\frac{\beta R}{1-\beta}$.\newline \indent Finally, let us consider the opposite case of a good enough distribution of returns, $(1/2)(y_{_H}+y_{_L})\geq R-(\lambda \theta /4)(y_{_H}-y_{_L})$. In this scenario, we know that the best decision on issuing risky assets is $s^*=\min \{a,\alpha_2^b\}$. Using this to solve for $\hat{m}$ and $\hat{a}$ we can write \begin{eqnarray} &&V^{3j}(m,b,s)= \kappa(m,b,s) \nonumber \\ &&+\max_{\hat{m},\hat{a}} \left\{-\phi \hat{m}-\psi\hat{a}+\beta \left[\frac{1}{2}V^{2L}(\hat{m}+\mu M,\hat{a}-\min\{\hat{a},\hat{\alpha}_2^b\},\min\{\hat{a},\hat{\alpha}_2^b\}) \right.\right.\nonumber \\ && \left.\left.+\frac{1}{2}V^{2H}(\hat{m}+\mu M,\hat{a}-\min\{\hat{a},\hat{\alpha}_2^b\},\min\{\hat{a},\hat{\alpha}_2^b\})\right]\right\}, \nonumber \end{eqnarray} for $j=H,L$, where $\hat{\alpha}_2^b=\frac{\hat{\phi}(\hat{\bar{m}}+\mu M)}{\hat{\psi}+\theta y_{_H}+(1-\theta)y_{_L}}$. If we are careful with the right expressions for $d_m(\hat{\bar{m}}+\mu M,\min \{a,\alpha_2^b\})$, $d_m(\hat{m}+\mu M,\hat{\bar{s}})$, $d_s(\hat{m}+\mu M,\hat{\bar{s}})$, and $d_s(\hat{\bar{m}}+\mu M,\min \{a,\alpha_2^b\})$, and after some algebra, we can arrive at a more useful objective function \begin{eqnarray} &&\max_{\hat{m},\hat{a}} \left\{-(\phi-\beta \hat{\phi}) \hat{m} -[\psi-\beta(\hat{\psi} +R)]\hat{a}\right. \nonumber \\ &&\left.+\left[\frac{1}{2}(y_{_L}+y_{_H})-R  + \frac{\lambda}{4}[\hat{\psi}+\theta y_{_H}+(1-\theta)y_{_L}-(\hat{\psi}+y_{_L})]\right]\min\left\{\hat{a},\frac{\hat{\phi}(\hat{\bar{m}}+\mu M)}{\hat{\psi}+\theta y_{_H}+(1-\theta)y_{_L}}\right\}\right.\nonumber \\ && \left. +\frac{\lambda \hat{\phi}}{4}\frac{[\hat{\psi}+(1-\theta) y_{_H}+\theta y_{_L}]}{\hat{\psi}+ y_{_L}}\min\left\{\hat{m},\frac{(\hat{\psi}+ y_{_H})(\hat{\psi}+ y_{_L})\hat{\bar{s}}}{[\hat{\psi}+(1-\theta) y_{_H}+\theta y_{_L}]\hat{\phi}}-\mu M\right\} \right. \nonumber \\ && \left. \frac{\lambda \hat{\phi}}{4}\min\left\{\hat{m},\frac{[\hat{\psi}+\theta y_{_H}+(1-\theta) y_{_L}]\hat{\bar{s}}}{\hat{\phi}}-\mu M\right\} \right\} \nonumber \end{eqnarray} \begin{eqnarray} \equiv \max_{\hat{m}}\{-c_1^m \hat{m}+c_2^m \min\{\hat{m},c_3^m\}-c_4^m \min\{\hat{m},c_5^m\}\} && \nonumber \\ + \max_{\hat{a}}\{-c_1^a â \hat{a}+c_2^a \min\{\hat{a},c_3^a\}\},&& \nonumber \end{eqnarray} where the definitions of $c_i^m$ and $c_k^a$, $i=1,...,5$, $k=1,2,3$, are obvious. The optimal solutions for $\hat{m}$ and $\hat{a}$ are fully described below \begin{eqnarray}\hat{m}^*=\left\{\begin{array}{l}0,\textrm{ if }c_2^m<c_1^m+c_4^m, \\ \in [0,c_3^m],\textrm{ if }c_2^m=c_1^m+c_4^m, \\c_3^m,\textrm{ if }c_2^m>c_1^m+c_4^m, \end{array}\right. &\textrm{and}&\hat{a}^*=\left\{\begin{array}{l}0,\textrm{ if }c_2^a<c_1^a, \\ c_3^a,\textrm{ if }c_2^a>c_1^a\\ \in [0,c_3^a],\textrm{ if }c_2^a=c_1^a,c_1^a\neq 0, \\\in [c_3^a,+\infty],\textrm{ if }c_1^a=0. \end{array}\right. \label{optimal mhat ahat appendix}\end{eqnarray} However, it is important to notice that with this optimal decisions market clearing in the DFM fails in the following sense. Within any given period the position of agents on money, $m$, and the real asset, $a$, are optimally chosen in the previous period, having taken into account the optimal choice for $s$ this period. Thus, within any given period we have $s=\min \left\{a,\alpha_2^b=\frac{\phi m}{\psi+\theta y_{_H}+(1-\theta) y_{_L}}\right\}$ and $m=c_3^m=\frac{(\psi+y_{_H})(\psi+y_{_L})s}{[\psi+(1-\theta) y_{_H}+\theta y_{_L}]\phi}$, which are clearly incompatible since we have shown that $\psi+\theta y_{_H}+(1-\theta) y_{_L}>\frac{(\psi+y_{_H})(\psi+y_{_L})}{[\psi+(1-\theta) y_{_H}+\theta y_{_L}]}$. Therefore, the only possible equilibrium has $c_1^a=0$, $\psi^*=\frac{\beta R}{1-\beta}$, $\hat{m}^*=c_3^m$, and \begin{eqnarray} \hat{a}^*=A=\frac{\hat{\phi}(\hat{\bar{m}}+\mu M)[\psi+(1-\theta) y_{_H}+\theta y_{_L}]}{(\hat{\psi}+y_{_H})(\hat{\psi}+y_{_L})}. \nonumber \end{eqnarray} Finally, we characterize the range of monetary policies consistent with equilibria under bargaining in the case of a good distribution of return. The optimal choice of money balances requires that \begin{eqnarray} \frac{\lambda \hat{\phi}}{4}\frac{[\hat{\psi}+(1-\theta)y_{_H}+\theta y_{_L}]}{(\hat{\psi}+y_{_L})} \geq \phi-\beta \hat{\phi}+\frac{\lambda \hat{\phi}}{4}. \nonumber\end{eqnarray} Now, if we use the fact that the only possible path for the price of the asset is $\psi=\beta R/(1-\beta)$ and rearrange the expression, we arrive at \begin{eqnarray}\mu \leq \beta -1 + \frac{\lambda}{4}\frac{(1-\theta)(y_{_H}-y_{_L})}{\frac{\beta R}{1-\beta}+y_{_L}}=\bar{\mu}_b,\label{upper bound mu}\end{eqnarray} where we have used the fact that in a steady state $\phi=\hat{\phi}(1+\mu)$. We conclude that there are many policies consistent with equilibrium, $\mu \in [\beta -1,\bar{\mu}_b]$. However, none of them affect asset prices, since $\psi$ in this framework always reflects the fundamental value of the asset.

\end{appendix}
\newpage
Figure 1: Securities Demand and Supply
\newpage
Figure 2: Optimal Choice of Securities
\newpage
Figure 3: Asset Price and Inflation


\begin{thebibliography}{99}
\bibitem{BR} Berentsen, A. and G. Rocheteau, 2003. Money and the Gains from Trade. International Economic Review 44 (1), 263-297.
\bibitem{DGP} Duffie, D., G\^{a}rleanu, N., and L.H. Pedersen, 2005. Over-the-Counter Markets. Econometrica 73, 1815-1847.
\bibitem{F} Ferraris, L., 2010. On the Complementarity of Money and Credit. European Economic Review 54 (5), 733-741.
\bibitem{FW} Ferraris, L. and M. Watanabe, 2011. Collateral Fluctuations in a Monetary Economy. Journal of Economic Theory, forthcoming.
\bibitem{GLS} Geromichalos, A., Licari, J.M., and J. Su\'{a}rez-Lled\'{o}, 2007. Monetary Policy and Asset Prices. Review of Economic Dynamics 10, no. 4, 761-779.
\bibitem{K} Kocherlakota, N., 1998. Money is Memory. Journal of Economic Theory 81, 232-251.
\bibitem{L} Lagos, R., 2011. Asset Prices, Liquidity, and Monetary Policy in an Exchange Economy. Journal of Money, Credit, and Banking, forthcoming.
\bibitem{LR} Lagos, R. and G. Rocheteau, 2009. Liquidity in Asset Markets with Search Frictions. Econometrica 73 (2), 403-426.
\bibitem{LRW} Lagos, R., Rocheteau, G., and P.O. Weill, 2011. Crises and Liquidity in Over-the-Counter Markets. Journal of Economic Theory, forthcoming.
\bibitem{LW} Lagos, R. and R. Wright, 2005. A Unified Framework for Monetary Theory and Policy Analysis. Journal of Political Economy 113, no. 3.
\bibitem{RW} Rocheteau, G. and R. Wright, 2005. Money in Search Equilibrium, in Competitive Equilibrium, and in Competitive Search Equilibrium. Econometrica 73, 175-202.

\end{thebibliography}
\end{document}